\documentclass[11pt]{article}

\usepackage{fullpage}
\usepackage{natbib}
\usepackage{ifthen}
\usepackage{xcolor}
\usepackage{amsmath,amssymb,amsthm}

\newtheorem{DefinitionNN}{Definition}
\newtheorem{LemmaNN}{Lemma}
\newtheorem{TheoremNN}{Theorem}

  \newcommand{\market}{M}
  \newcommand{\sample}{S}
  
  \newcommand{\val}{v}
  \newcommand{\vals}{{\mathbf \val}}
  \newcommand{\valsmi}[1][i]{{\mathbf \val}_{-#1}}
  \newcommand{\vali}[1][i]{{\val_{#1}}}
  
  \newcommand{\mals}{{\mathbf \val}^{\market}}
  
  \newcommand{\sals}{{\mathbf \val}^{\sample}}
  
  \newcommand{\eval}{\tilde{v}}
  \newcommand{\evals}{{\mathbf \eval}}
  \newcommand{\evali}[1][i]{{\eval_{#1}}}

  \DeclareMathOperator{\BSPEoperator}{BSPE}
  \newcommand{\BSPE}[1]{\BSPEoperator_{#1}}
  \DeclareMathOperator{\ICoperator}{IC}
  \newcommand{\IC}[1]{\ICoperator\ifthenelse{\not\equal{}{{#1}}}{^{{#1}}}{}}
  \DeclareMathOperator{\EFO}{EFO}
  \newcommand{\super}[1]{^{(#1)}}
  
  \newcommand{\prop}{p} 
  \newcommand{\pruin}{r} 

\newcommand{\prob}[2][]{\text{\bf Pr}\ifthenelse{\not\equal{}{#1}}{_{#1}}{}\!\left[#2\right]}
\newcommand{\expect}[2][]{\text{\bf E}\ifthenelse{\not\equal{}{#1}}{_{#1}}{}\!\left[#2\right]}
\newcommand{\given}{\,\middle|\,}

\begin{document}

\title{\Large The Biased Sampling Profit Extraction Auction}
\author{
  Bach Q. Ha\thanks{Department of Electrical Engineering and
  Computer Science, Northwestern University, Evanston, IL.  Email:
  {\texttt \{bach,hartline\}@u.northwestern.edu}.}\\
\and
  Jason D. Hartline\footnotemark[1]
}
\date{}
\maketitle

\begin{abstract}
 We give an auction for downward-closed environments that generalizes
 the random sampling profit extraction auction for digital goods of
 \citet{FGHK02}.  The mechanism divides the agents in to a market and
 a sample using a biased coin and attempts to extract the optimal
 revenue from the sample from the market.  The latter step is done
 with the downward-closed profit extractor of \citet{HH12}.  The
 auction is a $11$-approximation to the envy-free benchmark in
 downward-closed permutation environments.  This is an improvement on
 the previously best known results of $12.5$ for matroid and $30.4$
 for downward-closed permutation environments that are due to
 \citet{DHY12} and \citet{HH12}, respectively.
\end{abstract}

Economic mechanisms that are less dependent on the assumptions of the
environment are more likely to be relevant \citep[cf.][]{W85}.  The
area of {\em prior-free mechanism design} attempts to remove the
distributional assumption on agents while, at the same time,
guaranteeing a good approximation of the optimal revenue.

The performance of a prior-free mechanism is measured with respect to
a benchmark.  Recently, \citet{HY11} proposed the \emph{envy-free
  benchmark}, denoted by $\EFO(\vals)$ where $\vals$ is the valuation
vector of the agents.  This benchmark is the maximum revenue
attainable given that the allocation and payment vectors are
envy-free: no agent prefers another's outcome to her own.  A
downward-closed environment is one where given a feasible set of
agents, all subsets are feasible.  A permutation environment is one
where the agent identities are randomly permuted with respect to the
feasibility constraint.  In downward-closed permutation environments,
\citet{HY11} provide detailed justification for the approximation of
the envy-free benchmark.\footnote{For technical reasons the benchmark
  considered is $\EFO(\vals \super 2)$ where $\vals \super 2 =
  (\vali[2],\vali[2],\vali[3],\ldots,\vali[n])$ is same as $\vals$ but
  with the highest value lowered to the second highest value.}

The main approaches to prior-free auctions for digital goods
generalize to downward-closed permutation environments.  \citet{HY11}
generalized the random sampling auction; \citet{HH12} generalized the
consensus estimate profit extraction auction; and in the present paper
we generalize the random sampling profit extraction auction from
\citet{FGHK02}.  The random sampling profit extraction auction splits
the agents into a market and a sample, estimates the optimal profit
from the sample, and then attempts to extract that profit from the
market.  

\citet{HH12} give a profit extractor for the envy-free benchmark in
downward-closed permutation environments.  This profit extractor is
parameterized by a target valuation profile $\evals$ and on actual
valuation profile $\vals$ is able to extract at least the profit of
the envy-free benchmark $\EFO(\evals)$ when $\vals$ pointwise
dominates $\evals$, i.e., $\vali \geq \evali$ for all $i$, denoted $\vals \geq \evals$ where both $\vals$ and
$\evals$ are sorted in non-increasing order.

\begin{LemmaNN}[\citealp{HH12}]
\label{l:PErev}
 For downward-closed permutation environments there is a profit
 extractor parameterized by $\evals$ that obtains from $\vals$ at
 least the envy-free optimal revenue for $\evals$ if $\vals \geq
 \evals$ and otherwise rejects all agents.
\end{LemmaNN}

An unbiased partitioning of agents into a market $M$ and sample $S$
would be very unlikely to satisfy pointwise dominance $\mals \geq
\sals$ as necessary for the profit extractor of Lemma~\ref{l:PErev}
on $\mals$ to give revenue at least $\EFO(\sals)$.  On the other hand, a simple
probability of ruin analysis shows that a biased partitioning satisfies
the requisite pointwise dominance property with constant probability.

\begin{LemmaNN}
\label{l:props} 
For partitioning of $N$ into $S$ (with probability~$\prop < 1/2$) and $M$
(otherwise) satisfies
$\prob{\mals\not\geq\sals}\leq\tfrac{\prop}{1-\prop}$ and $\prob{\mals\not\geq\sals \given 1 \in M} \leq \big(\tfrac{\prop}{1-\prop}\big)^2$.
\end{LemmaNN}

\begin{proof}
  Consider the following infinite random walk on a straight line:
  starting from position $0$, with probability $\prop$, move backward
  one step; otherwise, move forward one step.  The position of this
  random walk describes precisely the difference between the number of
  agents in $\market$ and $\sample$, where positive value means
  $\market$ has more agents than $\sample$.  The event
  $\mals\not\geq\sals$ happens when there exists a time that $\market$
  has less agents than $\sample$.  Let $\pruin$ be the probability
  that the random walk eventually takes one step backward from the
  initial position, we have $\pruin=\prop+(1-\prop)\pruin^2$.  The
  first component is the probability of taking one step backward in
  the first step, and the second component is the probability of the
  first step being a forward step, then eventually take two steps
  backward.  Solving this equation for $\pruin \in (0,1)$ gives
  $\pruin=\prop/(1-\prop)$.  When we condition on $1 \in M$, our
  initial position is 1 not 0 and the probability of ruin is
  $\pruin^2$.  If we stop the random walk after finite number $n$ of
  steps, it only improves the probability of ruin.
\end{proof}

The random sampling profit extraction auction is formally given below
with a few modifications for improved performance.

\begin{DefinitionNN}[$\BSPE{\prop}$]
\label{mech:bspe}
The \emph{biased sampling profit extraction} auction parameterized by
$\prop<0.5$ works as follow.
\begin{enumerate}

 \item\label{mech:bspe-bs} Randomly assign each of the agents to one
   of three groups $A$, $B$, and $C$ independently with probabilities
   $p$, $p$, and $1-2p$, respectively.

 \item\label{mech:bspe-v1} Assume without loss of generality that of
   the highest valued agent in $A$ has value at least that of the
   highest valued agent in $B$.  Define the market $M = A \cup C$ and
   sample $S = B$.  (If this highest valued agent in $A$ wins in
   Step~\ref{mech:bspe-pe} and the second highest valued agent in $A
   \cup B$ is in $B$ increase her payment to this second highest
   value.)

 \item\label{mech:bspe-pad} Pad the valuation vectors of $\market$ and
   $\sample$ with $0$'s so that they are equal in length.  Let the
   padded vectors be $\mals$ and $\sals$ respectively.
 \item\label{mech:bspe-pe} Run the profit extractor parameterized by
   $\sals$ on the market $\market$.
 \item\label{mech:bspe-vcg} If all agents are rejected by the profit
   extractor and it is feasible to serve agent 1 (the highest valued
   agent over all), serve her and charge her $\vali[2]$.
\end{enumerate}
\end{DefinitionNN}

\begin{LemmaNN}[Incentive Compatibility]\label{l:IC}
  For all probabilities $\prop$, $\BSPE{\prop}$ is incentive compatible.
\end{LemmaNN}

\newcommand{\mech}[1][]{\mathcal{M}^{#1}}
\newcommand{\paymenti}[1][]{\tau_i^{#1}}

\begin{proof}
 Fix the partitioning of $A$, $B$, and $C$.  No agent in $M$ can
 change the definition of sets $M$ and $S$ without losing (thus
 obtaining zero utility).  No agent in $S$ can change the definition
 of sets $M$ and $S$ without obtaining a payment of at least her value
 (from the parenthetical in Step~\ref{mech:bspe-v1}, thus obtaining
 non-positive utility). Therefore no agent wants to manpulate the
 definition of $M$ and $S$.  For given $M$ and $S$ this mechanism is
 the profit extraction mechanism which is incentive compatible for
 fixed $M$ and $S$.  Only the highest valued agent would want to win
 in Step~\ref{mech:bspe-vcg}; furthermore, she cannot cause dominance
 to fail without lowering her bid (and forfeiting her status as the
 highest bidder).
\end{proof}


\newcommand{\safunc}{f}
\newcommand{\saSet}{N}
\newcommand{\saSubset}{M}
\newcommand{\saProp}{p}
\begin{LemmaNN}
\label{l:randomselection}
The envy-free benchmark $\EFO(\vals)$ for a random sample $S$ of $N$ with each element selected independently with probability $\prop$ satisfies $\expect{\EFO(\sals)} \geq \prop \EFO(\vals)$.
\end{LemmaNN}
\begin{proof}
Consider the envy-free optimal outcome for $\vals$.  Clearly if we
restrict attention only to agents in $S$ there is still no envy.
Therefore, $\EFO(\sals) \geq \EFO_S(\vals)$ where $\EFO_S(\vals)$ is
short-hand notation for the contribution from agents in $S$ to the
envy-free optimal revenue on $\vals$.  Of course, $\expect{\EFO_S(\vals)} = \prop \EFO(\vals)$.
\end{proof}

\begin{LemmaNN}
For any downward-closed permutation environment and any probability $\prop<0.5$,\footnote{The first part of this lemma is non-trivial only for $\prop < 0.38$.}
 \begin{enumerate}
  \item\label{l:approxvalsmi1} $\BSPE{\prop}$ approximates 
   $\EFO(\valsmi[1])$ to within a factor of $\prop-\big(\tfrac{\prop}
   {1-\prop}\big)^2$ where $\valsmi[1] = (\vali[2],\vali[3],\ldots,\vali[n])$. 
  \item\label{l:approxvali2} $\BSPE{\prop}$ approximates $\EFO(\vali[2])$ to within a factor of $\prop+(1-\prop)\prop^3$ when there are $n \geq 5$ agents.
\end{enumerate}
\end{LemmaNN}

\begin{proof}
 To show part~\ref{l:approxvalsmi1} of the lemma, we will focus on the revenue
 obtainable via the profit extraction step.  Lemma~\ref{l:PErev} says that
 we would obtain at least $\EFO(\sals)$ when $\mals\geq\sals$.  Thus the
 expected revenue is at least:
 \begin{align*}
  \expect{\BSPE{\prop}(\vals)}
  &\geq  \expect{\EFO(\sals)\mid\mals\geq\sals}\cdot\prob{\mals\geq\sals}\\
    & =  \expect{\EFO(\sals)} - \expect{\EFO(\sals)\mid\mals\not\geq\sals}\cdot\prob{\mals\not\geq\sals}\\
  &\geq  \prop\EFO(\valsmi[1])-\EFO(\valsmi[1])\cdot\prob{\mals\not\geq\sals}\\
    &=  \big[\prop-\big(\tfrac{\prop}{1-\prop}\big)^2\big]\cdot\EFO(\valsmi[1]).
 \end{align*}
 The second inequality warrants some explanation: the first term
 follows from applying Lemma \ref{l:randomselection} to $\valsmi[1]$,
 the second term follows from monotonicity of $\EFO$.

 
 To show part~\ref{l:approxvali2} of the lemma, we analize the event
 that player 2 is in the sample and the event that the market
 pointwise dominates the sample.  With probability $\prop$ agent 2 is
 in the sample and $\EFO(\sals) \geq \EFO(\vali[2])$.  If the market
 pointwise dominates the sample then the mechanism obtains this
 revenue; otherwise, the revenue from agent 1 via
 Step~\ref{mech:bspe-vcg} is at least $\EFO(\vali[2])$.  With
 probability $(1-\prop)$ agent 2 is in the market and the probability
 of $\mals \not\geq \sals$ (implying that the profit extractor fails)
 is at least $\prop^3$ by stepping backward three steps in a row
 (possible when there are more than $5$ agents); in this case again
 the revenue from agent 1 via Step~\ref{mech:bspe-vcg} is at least
 $\EFO(\vali[2])$.  To conclude, the revenue of the mechanism it at
 least:
 \begin{align*}
  \expect{\BSPE{\prop}(\vals)}
  &\geq  \big[\prop+(1-\prop)\prop^3\big] \cdot \EFO(\vali[2]). \qedhere
 \end{align*}

\end{proof}

\begin{TheoremNN}
 For any downward-closed permutation environment with probability
 $\prop<0.5$ and $n \geq 5$ agents,\footnote{For $n\leq 4$ agents the
   1-unit Vickrey auction is a 4-approximation to $\EFO(\vals \super
   2)$.} $\BSPE{\prop}$ approximates $\EFO(\vals \super 2)$ within a
 factor of $\tfrac{r_1+r_2}{r_1r_2}$ where $r_1=\prop-
 \big(\tfrac{\prop}{1-\prop}\big)^2$ and
 $r_2=\prop+(1-\prop)\prop^3$. This factor is minimized at $11$ when
 $\prop=0.26$.
\end{TheoremNN}

\begin{proof}
 $\EFO(\vali[2])+\EFO(\valsmi[1])\geq\EFO(\vals \super 2)$ due to
 subadditivity of $\EFO$ function as shown by \citet{HY11}. Combining with
 the above lemma, we have the desired ratio.
\end{proof}

\bibliographystyle{apalike}
\bibliography{bspe}

\end{document}